\documentclass{article}

\usepackage{fullpage}

\newtheorem{definition}{Definition}

\newtheorem{theorem}{Theorem}
\newtheorem{lemma}{Lemma}
\newtheorem{corollary}{Corollary}
\newenvironment{proof}{\paragraph{Proof:}}{\hfill$\square$}

\usepackage[ruled]{algorithm2e}

\usepackage{amsmath,amsfonts,amssymb,bbm} 
\SetArgSty{textrm}  
\SetAlFnt{\small}
\SetAlCapFnt{\small}
\SetAlCapNameFnt{\small}
\SetAlCapHSkip{0pt}
\IncMargin{-\parindent}

\usepackage{tikz}
\usetikzlibrary{patterns,snakes}
\tikzstyle{overbrace text style}=[font=\tiny, above, pos=.5, yshift=5pt]
\tikzstyle{overbrace style}=[decorate,decoration={brace,raise=5pt,amplitude=3pt}]

\newcount\Comments
\definecolor{darkgreen}{rgb}{0,0.7,0}
\newcommand{\kibitz}[2]{\ifnum\Comments=1\textcolor{#1}{#2}\fi}

\begin{document}


\markboth{Caragiannis et al.}{Truthful facility assignment with resource augmentation}

\title{Truthful Facility Assignment with Resource Augmentation: An Exact Analysis of Serial Dictatorship\thanks{Ioannis Caragiannis was partially supported by a Caratheodory research grant E.114 from the University of Patras. Aris Filos-Ratsikas was partially supported by the COST Action IC1205 on ``Computational Social Choice'' and by the ERC Advanced Grant 321171 (ALGAME). Aris Filos-Ratsikas, S\o ren Kristoffer Stiil Frederiksen, and Kristoffer Arnsfelt Hansen acknowledge support from the Danish National Research Foundation and The National Science Foundation of China (under the grant 61361136003) for the Sino-Danish Center for the Theory of Interactive Computation and from the Center for Research in Foundations of Electronic Markets (CFEM), supported by the Danish Strategic Research Council.}}

\author{Ioannis Caragiannis\footnote{
		University of Patras, Greece. E-mail: {\tt caragian@ceid.upatras.gr}}
	\and
	Aris Filos-Ratsikas\footnote{
		University of Oxford, United Kingdom. E-mail: {\tt aris.filos-ratsikas@cs.ox.ac.uk}}\\
	\and
	S\o ren Kristoffer Stiil Frederiksen\footnote{
		Aarhus University, Denmark. E-mail: {\tt sorensf@gmail.com}}
	\and
	Kristoffer Arnsfelt Hansen\footnote{
		Aarhus University, Denmark. E-mail: {\tt arnsfelt@cs.au.dk}}
	\and
	Zihan Tan\footnote{
		University of Chicago, Unites States. E-mail: {\tt zihantan@uchicago.edu}}
}
\date{} 
\maketitle
\begin{abstract}
	We study the \emph{truthful facility assignment} problem, where a set of agents with private most-preferred points on a metric space are assigned to facilities that lie on the metric space, under capacity constraints on the facilities. The goal is to produce such an assignment that minimizes the social cost, i.e., the total distance between the most-preferred points of the agents and their corresponding facilities in the assignment, under the constraint of truthfulness, which ensures that agents do not misreport their most-preferred points. 
	
	We propose a \emph{resource augmentation framework}, where a truthful mechanism is evaluated by its worst-case performance on an instance with enhanced facility capacities against the optimal mechanism on the same instance with the original capacities. We study a very well-known mechanism, Serial Dictatorship, and provide an exact analysis of its performance. 
Although Serial Dictatorship is a purely combinatorial mechanism, our analysis uses linear programming; a linear program expresses its greedy nature as well as the structure of the input, and finds the input instance that enforces the mechanism have its worst-case performance. Bounding the objective of the linear program using duality arguments allows us to compute tight bounds on the approximation ratio. Among other results, we prove that Serial Dictatorship has approximation ratio $g/(g-2)$ when the capacities are multiplied by any integer $g \geq 3$. Our results suggest that even a limited augmentation of the resources can have wondrous effects on the performance of the mechanism and in particular, the approximation ratio goes to $1$ as the augmentation factor becomes large. We complement our results with bounds on the approximation ratio of Random Serial Dictatorship, the randomized version of Serial Dictatorship, when there is no resource augmentation.
\end{abstract}






\newpage

\section{Introduction}

We study the \emph{facility assignment problem}, in which there is a set of agents and a set of \emph{facilities} with finite capacities; facilities are located on a metric space at points $F_i$ and each agent has a most-preferred point $A_i$, which is her private information. The goal is to produce an \emph{assignment} of agents to facilities, such that no capacity is exceeded and the sum of distances between agents and their assigned facilities, \emph{the social cost}, is minimized. A \emph{mechanism} is a function that elicits the points $A_i$ from the agents and outputs an assignment. We will be interested in \emph{truthful} mechanisms, i.e., mechanisms that do not incentivize agents to misreport their most-preferred locations and we will be aiming to find mechanisms that achieve a social cost as close as possible to that of the optimal assignment when applied to the true points $A_i$ of the agents.

Our setting has various applications such as assigning patients to personal GPs, vehicles to parking spots, children to schools and pretty much any matching environment where there is some notion of distance involved. Note that when being assigned to a personal GP, some patients might prefer to be assigned to someone closer to their house or workplace, so it is only natural to elicit their most-preferred points.

Our work falls under the umbrella of \emph{approximate mechanism design without money}, a term coined by Procaccia and Tennenholtz~\cite{PT:09} to describe problems where some objective function is optimized under the hard constraints imposed by the requirement of truthfulness. The standard measure of performance for truthful mechanisms is the \emph{approximation ratio}, which for our objective, is the worst-case ratio between the social cost of the truthful mechanism in question over the minimum social cost, calculated over all input instances of the problem.

However, it is arguably unfair to compare the performance of a mechanism that is severely limited by the requirement of truthfulness to that of an omnipotent mechanism that operates under no restrictions and has access to the real inputs of the agents, without giving the truthful mechanism any additional capabilities. This is even more evident in general settings, where strong impossibility results restrict the performance of all truthful mechanisms to be rather poor. The need for a departure from the worst-case approach has been often advocated in the literature, but the suggestions mainly involve some average case analysis or experimental evaluations. 

Instead, we will adopt a different approach, that has been made popular in the field of online algorithms and competitive analysis \cite{kalyanasundaram2000speed,sleator1985amortized}; the approach suggests enhancing the capabilities of the mechanism operating under some very limiting requirement (such as truthfulness or lack of information) before comparing to the optimal solution. Our main conceptual contribution is the adoption of a \emph{resource augmentation approach to approximate mechanism design}. In the resource augmentation framework, we evaluate the performance of a truthful mechanism on an input with additional resources, when compared to the optimal solutions for the set of original resources. For our problem, we consider the social cost achievable by a truthful mechanism on some input with augmented facility capacities against the optimal assignment under the original capacities given as input. 

More precisely, let $I$ be an input instance to the facility assignment problem and let $I_g$ be the same instance where each capacity has been multiplied by some integer constant $g$, that we call the \emph{augmentation factor}. Then, the \emph{approximation ratio with augmentation $g$} of a truthful mechanism $M$ is the worst-case ratio of the social cost achievable by $M$ on $I_g$ over the social cost of the optimal assignment on $I$, over all possible inputs of the problem. The idea is that if the ratio achievable by a mechanism with small augmentation is much better when compared to the standard approximation ratio, it might make sense to invest in additional resources. At the same time, such a result would imply that the set of ``bad'' instances in the worst-case analysis is rather pathological and not very likely to appear in practice.

To the best of our knowledge, this is the first time that such a resource augmentation framework has been employed in algorithmic mechanism design.

\subsection{Our results}

As our main contribution, we study the well-known truthful mechanisms for assignment problems, Serial Dictatorship (SD) and Random Serial Dictatorship (RSD). These are mechanisms of a greedy nature; SD fixes an arbitrary ordering of the agents and then assigns each agent to the facility closest to her most-preferred location $A_i$ from the set of facilities with leftover capacities. RSD is quite similar, but the ordering of agents is chosen uniformly at random from the set of all permutations of $n$ elements.

For SD, we provide an \emph{exact analysis}, obtaining tight bounds on the approximation ratio of the mechanism for all possible augmentation factors $g$. Specifically, we prove that when $n$ is the number of agents, while without any augmentation, the approximation ratio of SD is $2^n-1$, the approximation ratio with augmentation factor $g=2$ is exactly $\log (n+1)$ whereas for $g \geq 3$, the approximation ratio is $g/(g-2)$, i.e., a small constant. In particular, our results imply that as the augmentation factor becomes large, the approximation ratio of SD with augmentation goes to $1$. Our results for SD improve and extend some results in the literature of online algorithms \cite{kalyanasundaram2000online}, as we will explain in the next subsection.

To prove the approximation ratios for all augmentation factors, we use an interesting technique based on linear programming. Specifically, we first provide a directed graph interpretation of the assignment produced by SD and the optimal assignment, and then prove that the worst-case instances appear on $g$-trees, i.e., trees where (practically) every vertex has exactly $g$ successors. Then, we formulate the problem of calculating the worst ratio on such trees as a linear program and bound the ratio by obtaining feasible solutions to its dual. A feasible solution to the dual can be seen as a ``path covering'' of the assignment graph and we obtain the bounds by constructing appropriate path coverings of low cost. 

We also consider randomized mechanisms and the very well-known Random Serial Dictatorship mechanism. We prove that for augmentation factor $1$ (i.e., no resource augmentation), the approximation ratio of the mechanism is between $n^{0.26}$ and $n$; the result suggests that even a small augmentation ($g=2$) is a more powerful tool than randomization.
Again, as we will explain in the next subsection, this result has corollaries in the field of online algorithms. 

\subsection{Related Work}

Assignment problems are central in the literature of economics and computer science; the literature on one-sided matchings dates back to the seminal paper by Hylland and Zeckhauser \cite{HZ:79} and includes many very influential papers \cite{BM:01,SVE:99} in economics as well as a rich recent literature in computer science \cite{GC:10,anshelevich2010matching,CS14}. Serial Dictatorships (or their randomized counterparts) have been in the focus of much of this literature, mainly due to their simplicity and the fragile nature of truthfulness, which makes it quite hard to construct more involved truthful mechanisms. In a celebrated result, Svensson \cite{SVE:99} characterized a large class of truthful mechanisms by serial dictatorships. Random Serial Dictatorship has also been extensively studied \cite{krysta2014size, abdulkadirouglu1998random} and recently it was proven \cite{FFZ:14} that is asymptotically the best truthful mechanism for one-sided matchings under the general cardinal preference domain.  

The facility assignment problem can be interpreted as a matching problem;
somewhat surprisingly, matching problems in metric spaces have only recently been considered in the mechanism design literature. Emek et al. \cite{emek2015price} study a setting very closely related to ours, where the goal is to find matchings on metric spaces, but they are interested in how well a mechanism that produces a stable matching can approximate the cost of the optimal matching. In a conceptually similar work, Anshelevich and Shreyas \cite{anshelevich2015blind} study the performance of \emph{ordinal} matching mechanisms on metric spaces, when the limitation is the  lack of information. The fundamental difference between those works and ours is that we consider truthful mechanisms and bound their performance due to the truthfulness requirement; to the best of our knowledge, this is the first time where truthful mechanisms have been considered in a matching setting with metric preferences. 
Another difference between our work and the aforementioned papers is that they do not consider resource augmentation and only bound the performance of mechanisms on the same set of resources.\footnote{With the exception of the bi-criteria result in \cite{anshelevich2015blind}.} 
However, given the generality of the augmentation framework, the same idea could be applied to their settings. 
In that sense, our paper proposes a \emph{resource augmentation approach to algorithmic mechanism design} that could be adopted in most resource allocation and assignment settings.

As we mentioned earlier, the idea of resource augmentation was popularized by the field of online algorithms and competitive analysis and is tightly related to the literature on \emph{weak adversaries} where an online competitive algorithm is compared to the adversary that uses a smaller number of resources. The idea for this approach originated in the seminal paper by Sleator and Tarjan \cite{sleator1985amortized} and has been adopted by others ever since \cite{koutsoupias1999weak,young1994thek}; the term ``resource augmentation'' was explicitly introduced by Kalyanasundaram and Pruhs~\cite{kalyanasundaram2000speed}.

Most closely related to our problem is the \emph{online transportation problem} \cite{kalyanasundaram2000online,meyerson2006randomized} (also known as the minimum online metric bipartite matching). In particular, as we explain in Section \ref{sec:onlinetransportation} of the Appendix, results about the greedy algorithm in the online transportation problem imply bounds for the facility assignment problem. However, contrary to \cite{kalyanasundaram2000online}, our analysis is \emph{exact}, i.e. our results involve no asymptotics. Furthermore, compared to the related result in \cite{kalyanasundaram2000online}, we remark that our analysis is substantially different due to the use of linear programming; our primal-dual technique could be applicable for greedy assignment mechanisms on other resource augmentation settings, beyond the specific problem.

Compared to the related result in \cite{kalyanasundaram2000online}, we remark that our analysis is substantially different due to the use of linear programming. This technique for the analysis of purely combinatorial algorithms has found applications in many different contexts such as facility location \cite{JMM+03}, set cover \cite{ACK09}, online matching \cite{MY11}, maximum directed cut \cite{FJ15}, wavelength routing \cite{C09}, and revenue optimization \cite{ACV15}. Like in our case, these techniques usually lead to tight analysis. Also note that while the connection between SD and RSD and the greedy algorithm for the online transportation problem is straightforward, the two problems are fundamentally different and hence non-greedy online competitive algorithms do not imply any bounds for our setting and non-serial truthful mechanisms do not imply any bounds for the online setting. 

Finally, there is some resemblance between our problem and the facility location problem \cite{PT:09} that has been studied extensively in the literature of approximate mechanism design, in the sense that in both settings, agents specify their most preferred positions on a metric space. Note that the settings are fundamentally different however, since in the facility location problem, the task is to identify the appropriate point to locate a facility whereas in our setting, facilities are already in place and we are looking for an assignment of agents to them. 

\section{Preliminaries}

In the \emph{facility assignment} problem, there is a set $N=\{1,\ldots,n\}$ of agents and a set $M=\{1,\ldots,m\}$ of \emph{facilities}, where agents and facilities are located on a 
metric space, equipped with a distance function $d$. Each facility has a \emph{capacity} $c_i \in \mathbb{N_+}$, which is the number of agents that the facility can accommodate. We assume that $\sum_{i=1}^{m} c_i \geq n$, i.e., all agents can be accommodated by some facility. Each agent has a most preferred position $A_i$ on the space and his cost $d_i(j)$ from facility $j$ is the distance $d(A_i,F_j)$ between $A_i$ and the position $F_j$ of the facility.  Let $A=(A_1,\ldots,A_n)$ be a vector of preferred positions and call it a \emph{location profile}. Let $F=(F_1,\ldots,F_m)$ be the corresponding set of points of the facilities. A pair of agents' most preferred points and facility points $(A,F)$ is called an \emph{instance} of the facility assignment problem and is denoted by $I$.

The locations of the facilities are known but the location profiles are not known; agents are asked to report them to a central planner, who then decides on an \emph{assignment} $S$, i.e., a pairing of agents and facilities such that no agent is assigned to more than one facility and no facility capacity is exceeded. Let $S_i$ be the restriction of the assignment to the $i$'th coordinate, i.e., the facility to which agent $i$ is assigned in $S$ and let $\mathcal{S}$ be the set of all assignments. The \emph{social cost} of an assignment $S$ on input $I$ is the sum of the agents' costs from their facilities assigned by $S$ i.e., $\sum_{i=1}^n d_i(S_i)$.  A deterministic mechanism maps instances to assignments whereas a randomized mechanism maps instances to probability distributions over assignments. 

A mechanism is \emph{truthful} if no agent has an incentive to misreport his most preferred location. 
Formally, this is guaranteed when for every location profile $A$, any report $A_i'$, and any reports $A_{-i}$ of all agents besides agent $i$, it holds that 
$d_i(S_i)  \geq d_i(S'_i)$, where $S=M(I)$ and $S'=M(I')$, with $I=(A,F)$ and $I'=((A_i',A_{-i}),F)$.
For randomized mechanisms, the corresponding notion is \emph{truthfulness-in-expectation}, where an agent can not decrease her expected distance from the assigned facilities by deviating, i.e., it holds that $\mathbb{E}_{S \sim D} [d_i(S_i)] \geq \mathbb{E}_{S \sim D'} [d_i(S_i)]$, where $D$ and $D'$ are the probability distributions output by the mechanism on inputs $I$ and $I'$ respectively. A stronger notion of truthfulness for randomized mechanisms is that of \emph{universal truthfulness}, which guarantees that for every realization of randomness, there will not be any agent with an incentive to deviate. Alternatively, one can view a universally truthful mechanism as a mechanism that runs a deterministic truthful mechanism at random, according to some distribution.

As our main conceptual contribution, we will consider a \emph{resource augmentation} framework where the minimum social cost of any assignment will be compared with the social cost achievable by a mechanism on a location profile with augmented facility capacities. Given an instance $I$, we will use the term \emph{$g$-augmented instance} to refer to an instance of the problem where the input is $I$ and the facility of each capacity has been multiplied by $g$. We will denote that instance by $I_g$ and we will call $g$ the \emph{augmentation factor} of $I$. For example, when $g=2$, we will compare the minimum social cost with the social cost of a mechanism on the same inputs but with double capacities.

For the facility assignment problem, the optimal mechanism computes a minimum cost matching (which can be computed using an algorithm for maximum weight bipartite matching) and it can be easily shown that it is not truthful; in order to achieve truthfulness, we have to output suboptimal solutions. As performance measure, we define the \emph{approximation ratio with augmentation} of a mechanism $M$ as
\begin{eqnarray*}
	ratio_g(M) = \sup_{I} \frac{SC_M(I_g)}{SC_{OPT}(I)}
\end{eqnarray*}
\noindent where $SC_M(I_g)=\sum_{i=1}^n d_i(M(I_g)_i)$ is the social cost of the assignment produced by mechanism $M$ on input instance $I$ with augmentation factor $g$ and $SC_{OPT}(I)$ is the minimum social cost of any assignment on $I$ i.e., $SC_{OPT}(I) = \min_{S \in \cal{S}} \sum_{i=1}^n d_i(S_i)$. For randomized mechanisms, the definitions involve the expected social cost and are very similar.
Obviously, if we set $g=1$, we obtain the standard notion of the approximation ratio for truthful mechanisms \cite{PT:09}. For consistency with the literature, we will denote $ratio_1(M)$ by $ratio(M)$.

We will be interested in two natural truthful mechanisms that assign agents to facilities in a greedy nature. A \emph{serial dictatorship} (SD) is a mechanism that first fixes an ordering of the agents and then assigns each agent to his most preferred facility, from the set of facilities with non-zero residual capacities. Its randomized counterpart, \emph{Random Serial Dictatorship} (RSD), is the mechanism that first fixes the ordering of agents uniformly at random and then assigns them to their favorite facilities that still have capacities left. In other words, RSD runs one of the $n!$ possible serial dictatorships uniformly at random and hence it is universally truthful.

\section{Approximation Guarantees for Serial Dictatorships} \label{sec:analysis}

In this section we provide our main results, the upper bounds on the approximation ratio with augmentation of Serial Dictatorship, for all possible augmentation factors. The results can be summarized in the following theorem. In Section \ref{sec:lowerbounds}, we provide instances for which the bounds proven here are tight.
At the end of the section, we also consider Random Serial Dictatorship, and provide an upper bound on the approximation ratio of the mechanism when there is no resource augmentation.

\begin{theorem}\label{thm:SD}
	The approximation ratio of SD with augmentation factor $g$ in facility assignment instances with $n$ agents is 
	\begin{enumerate}
	\item $ratio(SD)\leq 2^{n}-1$, 
	\item $ratio_2(SD)\leq\log (n+1)$,
	\item $ratio_g(SD) \leq \frac{g}{g-2}$ when $g \geq 3$.
	\end{enumerate}
\end{theorem} 

Before we proceed, we would like to point out that Statement 1 in Theorem \ref{thm:SD} can be obtained as a corollary of the results in the literature of the online transportation problem. Specifically, it can be obtained as a corollary of Theorem 2.5 in \cite{kalyanasundaram1993online}. However, for completeness, we will reprove Statement 1 as part of our more general framework.

In order to do that, we first need to introduce a different interpretation of the assignment produced by SD and the optimal assignment, in terms of a directed graph. Let us begin with a high-level roadmap of the proof of Theorem \ref{thm:SD}.
\begin{enumerate}
	\item We show how to represent an instance of facility assignment together with an optimal solution and a solution computed by the SD mechanism as a directed graph and argue that the instances in which the SD mechanism has the worst approximation ratio are specifically structured as directed trees.
	\item We observe that the cost of the SD mechanism in these instances is upper-bounded by the objective value of a maximization linear program defined over the corresponding directed trees.
	\item We use duality to upper-bound the objective value of this LP by the value of a feasible solution for the dual LP. This reveals a direct relation of the approximation ratio of the SD mechanism to a graph-theoretic quantity defined on a directed tree, which we call the cost of a path covering.
	\item Our last step is to prove bounds on this quantity; these might be of independent interest and could find applications in other contexts.
\end{enumerate}  

Consider an instance $I$ of facility assignment. Recall the interpretation of the problem as a metric bipartite matching and note that without loss of generality, each facility can be assumed to have capacity $1$ and $m\geq n$. Unless otherwise specified, agents and facilities are identified by the integers in $[n]$ and $[m]$, respectively. 

Now, let $O$ be any assignment on input ${I}$, and let $S$ be an assignment returned by the SD mechanism when applied on the instance $I_g$ (where each facility has capacity $g$). We use a directed graph to represent the triplet $I$, $O$, and $S$ as follows. The graph has a node for each facility. Each directed edge corresponds to an agent. A directed edge from a node corresponding to facility $j_1$ to a node corresponding to facility $j_2$ indicates that the agent corresponding to the edge is assigned to facility $j_1$ in $O$ and facility $j_2$ in $S$. Observe that there is at most one edge outgoing from each node; this edge corresponds to the agent that is assigned to the facility corresponding to the node in solution $O$. Furthermore, a node may have up to $g$ incoming edges, corresponding to agents assigned to the facility by the SD mechanism.

Representations as directed $g$-trees are of particular importance. A \emph{directed $g$-tree $T$} is an acyclic directed graph that has a root node $r$ of in-degree $1$ and out-degree $0$, leaves with in-degree $0$ and out-degree $1$, and intermediate nodes with in-degree $g$ and out-degree $1$. We now show that it suffices to restrict our attention to directed $g$-trees as graph representations of instances in which the SD mechanism achieves its worst performance.

\begin{lemma}\label{lem:gtreesenough}
	Given a instance $I$ with $n$ agents, an optimal solution $O$ for $I$ and a solution $S$ consistent with the SD mechanism when applied to instance $I_g$, there is another instance $I'$ with at most $n$ agents, with an optimal solution $O'$ and a solution $S'$ consistent with the application of the SD mechanism on the instance $I_g$ such that the representation graph of the triplet $(I',O',S')$ is a directed $g$-tree and such that \[\frac{\mbox{cost}(S,I_g)}{\mbox{cost}(O,I)}\leq \frac{\mbox{cost}(S',I'_g)}{\mbox{cost}(O',I')}.\]
\end{lemma}

\begin{proof}
	Let $o_i$ and $s_i$ denote the facility to which agent $i$ is connected in assignments $O$ and $S$, respectively. We say that agent $i$ is {\em optimal} if $o_i=s_i$. We say that agent $i$ is {\em greedy} if $s_i\not=o_i$ and less than $g$ agents are assigned to facility $o_i$ when SD decides the assignment of agent $i$. This means that $d(A_i,F_{s_i}) \leq d(A_i,F_{o_i})$. We say that agent $i$ is {\em blocked} if $g$ agents are already assigned to facility $o_i$ when SD decides the assignment of agent $i$. 
	
	Starting from $({\cal I}, O, S)$, we construct a new triplet $({\cal I}', O', S')$ as follows: 
	\begin{itemize}
		\item First, we remove all optimal agents. This corresponds to removing loops from the representation graph. 
		\item Then, we repeat the following process as long as there exists a blocked agent $i$ that is connected under $S$ to a facility $j$ that is the optimal facility of a greedy agent. In this case, we introduce a new facility $j'$ at point $F_{j'}$ such that $d(A_i,F_{j'})=d(A_i,F_j)$ and $d(F_{j'},X)=d(A_i,F_{j'})+d(A_i,X)$ for every other point $X$ of the space. The second equality guarantees that the set of all points corresponding to locations of agents and facilities that have survived and the newly introduced point $F_{j'}$ is a metric. This can easily be achieved by placing the new facility $j'$ such that it coincides with $j$ on the metric space. We assign agent $i$ to facility $j'$ instead of $j$; by the first equality above, this is consistent to the definition of the SD mechanism. In the representation graph, this step adds a new node corresponding to the new facility $j'$ and modifies the directed edge corresponding to blocked agent $i$ so that it is directed to the new node. 
		\item Then, we remove all greedy agents that are not connected under $S$ to optimal facilities of blocked agents together with their optimal facilities. 
		\item Then, for each facility $j$ that is used by $t\geq 2$ agents $i_1$, $i_2$, ..., $i_t$ in $S$ but is not used by any agent in $O$, we remove facility $j$ and introduce $t$ new facilities $j_1$, $j_2$, ..., $j_t$ such that $d(A_{i_k},j_k)=d(A_{i_k},j)$ for $k=1, ..., t$ and $d(X,j_k)=d(X,A_{i_k})+d(A_{i_k},j_k)$ for every other point $X$ of the space. Again, the second equality guarantees that the set of all points corresponding to locations of agents and facilities that have survived and the newly introduced points $F_{j_1}$, ..., $F_{j_t}$ is a metric. For $k=1, ..., t$, we assign agent $i_k$ to facility $j_k$; by the first equality above, this is consistent to the definition of the SD mechanism. In the representation graph, this step adds $t$ nodes corresponding to the new facilities $j_1$, ..., $j_t$ and, for $k=1, ..., t$, it modifies the directed edge corresponding to blocked agent $i_k$ so that it is directed to the new node $j_k$, and removes node corresponding to facility $j$. 
		\item Finally, we remove any facility that is not used by any of the non-removed agents in any of the two solutions. 
	\end{itemize}
	We denote by ${\cal I}'$ the resulting instance and by $O'$ the restriction of $O$ to the survived agents. Also, $S'$ is the assignment obtained by the modification of $S$ and considering the survived agents only. We remark that the representation graph of $({\cal I}',O',S')$ is a forest of directed $g$-trees. Indeed, the optimal facility of a greedy agent is not used by any agent in $S'$; the corresponding node is a leaf in the representation graph. Now, assume that the representation graph contains a directed cycle; this should consist of directed edges corresponding to blocked agents. By the definition above, this would mean that, for every agent $j$ in this cycle, the assignment of all agents that were assigned by the SD mechanism to the optimal facility $o_j$ took place before the assignment of agent $j$ to a facility; this yields a contradiction and no such cycle exists. The optimal facility of a blocked agent has out-degree $1$ and in-degree $g$. Nodes with zero out-degree have degree exactly $1$; these are nodes corresponding to the newly added facilities and serve as roots of the directed $g$-trees.
	
	Let $R$ be the set of (greedy and optimal) agents removed and observe that $d(A_i,F_{s_i})\leq d(A_i,F_{o_i})$ for each such agent $i\in R$. Hence, it is 
	
	\begin{eqnarray*}
		\frac{\mbox{cost}(S,{\cal I}_g)}{\mbox{cost}(O,{\cal I})} &=& \frac{\sum_{i\in [n]}{d(A_i,F_{s_i})}}{\sum_{i\in [n]}{d(A_i,F_{o_i})}}
		\leq  \frac{\sum_{i\in [n]}{d(A_i,F_{s_i})}-\sum_{i\in R}{d(A_i,F_{s_i})}}{\sum_{i\in [n]\setminus R}{d(A_i,F_{o_i})}-\sum_{i\in R}{d(A_i,F_{o_i})}}\\
		&=&\frac{\sum_{i\in [n]\setminus R}{d(A_i,F_{s'_i})}}{\sum_{i\in [n]\setminus R}{d(A_i,F_{o'_i})}}.
	\end{eqnarray*} 
	
	Clearly, if the representation of triplet ${\cal I}', O', S'$ consists of more than one $g$-trees, there is an instance ${\cal I}''$ and assignments $O''$ and $S''$ corresponding to the restriction of $({\cal I}', O', S')$ in one of the $g$-trees which satisfies $\frac{\mbox{cost}(S,{\cal I}_g)}{\mbox{cost}(O,{\cal I})} \leq \frac{\mbox{cost}(S'',{\cal I}''_g)}{\mbox{cost}(O'',{\cal I}'')}$. If $O''$ is indeed an optimal solution for instance ${\cal I}''$, the proof is complete. Otherwise, we repeat the whole process using instance ${\cal I}''$ as ${\cal I}$, solution $O$ to be the optimal solution for instance ${\cal I}''$, and the SD solution $S''$ until the solution $O''$ obtained is optimal for the $g$-tree instance obtained at the final step (this condition will eventually be satisfied as the optimal cost decreases in each application of the process). By setting $\tilde{\cal I}={\cal I}''$, $\tilde{O}=O''$, and $\tilde{S}=S''$ will then yield the triplet with the desired characteristics.
\end{proof}

So, in the following, we will focus on triplets $({\cal I},O,S)$ of a facility assignment instance ${\cal I}$ with at most $n$ agents, with an optimal solution $O$, and with an SD solution $S$ for instance ${\cal I}_g$ that have a graph representations as a directed $g$-tree $T$. Below, we use ${\cal P}$ to denote the set of all paths that originate from leaves. Given an edge $e$ of a $g$-tree, we use ${\cal P}_e$ (respectively, $\tilde{\cal P}_e$) to denote the set of all paths that originate from a leaf and cross (respectively, terminate with) edge $e$. We always use $e_r$ to denote the edge incident to the root of a $g$-tree.

Our next observation is that $\mbox{cost}(S,{\cal I}_g)$ is upper-bounded by the objective value of the following linear program.
\begin{eqnarray*}
	\mbox{maximize} & & \sum_{e\in T}{z_e}\\
	\mbox{subject to:} & & z_e-\sum_{a\in p\setminus\{e\}}{z_a} \leq \sum_{a\in p}{d(A_a,F_{o_a})}, e\in T, p\in \tilde{\cal P}_e\\
	& & z_e \geq 0, e\in T
\end{eqnarray*}
To see why, interpret variable $z_e$ as the distance of agent corresponding to edge $e$ of $T$ to the facility it is connected to under assignment $S$. Then, clearly, the objective $\sum_{e\in T}{z_e}$ represents $\mbox{cost}(S,{\cal I}_g)$. Now, how high can $\mbox{cost}(S,{\cal I})$ be? The LP essentially answers this question (partially, becauses it does not use all constraints of the SD mechanism but sufficiently for our purposes). In particular, the LP takes into account the fact that the distance of agent $e$ to the facility to which it is connected in $S$ is not higher than the distance from the agent to any leaf facility in its subtree; this follows by the definition of the SD mechanism since leaf facilities are by definition available throughout the execution of the SD mechanism.
Indeed, consider agent $e$ and a path $p\in \tilde{P}_e$. Since agent $e$ is connected to facility $s_e$ under SD and not to the facility corresponding to the leaf from which path $p$ originates from, this means that the distance $d(A_e,F_{s_e})$ is not higher than the distance of $A_e$ from the location of the facility corresponding to that leaf. Since $d$ is a metric, this distance is at most $d(A_e,F_{o_e})+\sum_{a\in p\setminus\{e\}}{d(F_{s_a},F_{o_a})} \leq d(A_e,F_{o_e})+\sum_{a\in p\setminus\{e\}}{(d(A_a,F_{s_a})+d(A_a,F_{o_a}))}$. So, the constraint associated with path $p\in \tilde{P}_e$ in the LP captures the inequality $d(A_e,F_{s_e})\leq d(A_e,F_{o_e})+\sum_{a\in p\setminus\{e\}}{d(F_{s_a},F_{o_a})} \leq d(A_e,F_{o_e})+\sum_{a\in p\setminus\{e\}}{(d(A_a,F_{s_a})+d(A_a,F_{o_a}))}$, by replacing $d(A_e,F_{s_e})$ with $z_e$ and $d(A_a,F_{s_a})$ with $z_a$ and rearranging the terms.

By duality, the cost $\mbox{cost}({\cal I},S)$ of solution $S$ is upper-bounded by the objective value of the dual linear program, defined as follows:
\begin{eqnarray*}
	\mbox{minimize} & & \sum_{p\in {\cal P}}{x_p\sum_{e\in p}{d(A_e,F_{o_e})}}\\
	\mbox{subject to:} & & \sum_{p\in {\cal P}_{e_r}}{x_p} \geq 1\\
	& & \sum_{p\in \tilde{\cal P}_{e}}{x_p}-\sum_{p\in {\cal P}_e\setminus \tilde{\cal P}_e}{x_p} \geq 1, e\in T\\
	& & x_p \geq 0, p\in {\cal P}
\end{eqnarray*}
Actually, for any feasible solution $x$ of the dual LP, the quantity $\sum_{p\in {\cal P}}{x_p\sum_{e\in p}{d(A_e,F_{o_e})}}$ is an upper bound on $\mbox{cost}(S,{\cal I}_g)$. We will refer to any assignment $x$ over the paths of ${\cal P}$ that satisfies the constraints of the dual LP as a {\em path covering} of the directed $g$-tree $T$ and will denote its cost by $c(x)=\max_{e\in T}{\sum_{p\in {\cal P}_e}{x_p}}$. We repeat these definitions for clarity:

\begin{definition}
	Let $T$ be a directed tree. A function $x:{\cal P}\rightarrow \mathbb{R}^+$ is called a path covering of $T$ if the following conditions hold:
	\begin{itemize}
		\item $\sum_{p\in {\cal P}_{e_r}}{x_p}\geq 1$ for the edge $e_r$ incident to the root of $T$;
		\item $\sum_{p\in \tilde{\cal P}_e}{x_p}-\sum_{p\in {\cal P}_e \cap {\cal P}_{f}}{x_p}\geq 1$ if $e\not= e_r$ and $f$ denotes the parent edge of $e$.
	\end{itemize}
	The cost $c(x)$ of $x$ is equal to $\max_{e\in T}{\sum_{p\in {\cal P}_e}{x_p}}$.
\end{definition}

\begin{lemma}\label{lem:pathcovercost}
	Let $g\geq 2$ be an integer, ${\cal I}$ be a facility assignment instance with an optimal solution $O$, $S$ be a solution of the SD mechanism when applied on instance ${\cal I}_g$, so that the triplet $({\cal I},O,S)$ is represented as a directed $g$-tree $T$ which has a path covering $x$. Then, $\mbox{cost}(S,{\cal I}_g) \leq c(x)\cdot \mbox{cost}(O,{\cal I})$.
\end{lemma}
\begin{proof}
	Using the interpretation of the variables of the primal LP, duality, and the definition of the cost of path covering $x$, we have that 
	\begin{eqnarray*}
		\mbox{cost}(S,{\cal I}_g) &=& \sum_{e\in T}{z_e}
		\leq  \sum_{p\in {\cal P}}{x_p\sum_{e\in p}{d(A_e,F_{o_e})}}
		= \sum_{e\in T}{d(A_e,F_{o_e}) \cdot \sum_{p\in {\cal P}_e}{x_p}}\\
		&\leq & c(x) \cdot \sum_{e\in T}{d(A_e,F_{o_e})}
		= c(x)\cdot \mbox{cost}(O,{\cal I})
	\end{eqnarray*}
	as desired.
\end{proof}

In order to establish the upper bounds in Theorem \ref{thm:SD}, it remains to show that path coverings with low cost do exist; this is what we do in the next three lemmas. We start with the Lemma for no augmentation.

\begin{lemma}\label{lem:g1}
	Let $T$ be a $1$-tree. Then, there is a path covering of $T$ of cost $2^{n}-1$.
\end{lemma}	

\begin{proof}
First, observe that a directed $1$-tree consists of a single branch, where the first node (the leaf) has out-degree $1$ and in-degree zero, the last node (the root) has in-degree $1$ and out-degree $0$ and all other nodes have precisely one incoming edge and one outgoing edge. Therefore, for each edge $e$ in the tree, the set $\tilde{P}_{e}$ consists of a single path, that we will denote by $p_e$ and the set $P_e \cup P_f$, where $f$ is the parent edge of $e$, it holds that $P_e \cup P_f=P_{e_r}-\{p_e\}$, i.e., the set contains all the paths from the leaf to the root, except for path $p_e$ that ends at edge $e$. 

Let $m$ be the index of the facility corresponding to the leaf of the tree and let $i$ be the index of the facility with an incoming edge from facility index by $i+1$; observe that the root facility has index $1$. Let $e_i$ be the edge originating from facility $i+1$ to facility $i$. Now for every path $p_{e_i}$, let $x_{p_{e_i}}=2^i$. This is a complete assignment to all paths, since every path originates from the single leaf, and ends at either the root or some intermediate node. It is not hard to see that the assignment is a path covering, since $\sum_{p \in P_{e_r}} \geq 1$ and for every path $p_{e_{i}}$ it holds that $x_{p_{e_i}} \geq \sum_{j=1}^{i-1} x_{p_{e_j}}$. The cost of the path finding is $c(x) = \sum_{i=1}^{m-1} x_{p_{e_i}} = \sum_{i=1}^{m-1} 2^i = 2^m -1$ which is at most $2^n -1$, since $n\geq m$.
\end{proof}

In the following, we identify path coverings of low cost for the case of $g \geq 3$ and $g=2$. The next two lemmas complete the part of Theorem \ref{thm:SD} that regards the upper bounds.

\begin{lemma}\label{lem:g3}
	Let $g\geq 3$ be an integer and $T$ be a $g$-tree. Then, there is a path covering of $T$ of cost $\frac{g}{g-2}$.
\end{lemma}

\begin{proof}
	We prove the lemma using the following assignment $x$: for every path $p$ of length $\ell$, we set $x_p=\frac{1}{g-2}g^{2-\ell}$ if it contains and edge that is adjacent to the root and $x_p=\frac{g-1}{g-2}g^{1-\ell}$ otherwise.
	
	We will first show that $\sum_{p\in{\cal P}_e}{x_p}=\frac{g}{g-2}$ for every edge $e$ using induction. We will do so by visiting the edges in a bottom-up manner (i.e., an edge will be visited only after its child-edges have been visited) and prove that the equality for edge $e$ using the information that the equality holds for its child-edges. As the basis of our induction, consider an edge $e$ that is adjacent to a leaf at depth $\ell\geq 1$ from the root. If $\ell=1$, this means that the tree consists of a single edge and there is a single path $p$ with $x_p=\frac{g}{g-2}$. If $\ell\geq 2$, then the paths that contain edge $e$ are those who end at each ancestor of the leaf adjacent to $e$. Hence,
	\begin{eqnarray*}
		\sum_{p\in {\cal P}_e}{x_p}&=&\sum_{i=1}^{\ell-1}{\frac{g-1}{g-2}g^{1-i}}+\frac{1}{g-2}g^{2-\ell}
		= \frac{g}{g-2}.
	\end{eqnarray*}
	Now, let us focus on a non-leaf edge $e$ and assume that $\sum_{p\in {\cal P}_{e_i}}{x_p}=\frac{g}{g-2}$ for each child-edge $e_i$ (for $i\in [g]$) of $e$ (this is the induction hypothesis). Let $u$ be the node to which edges $e$ and $e_i$ with $i\in [g]$ are incident. The set of paths in ${\cal P}_e$ consists of the following disjoint sets of paths: for each edge $e_i$ and for each path $p\in \tilde{\cal P}_{e_i}$, set ${\cal P}_e$ contains all super-paths of $p$, i.e., paths originating from the leaf-node reached by $p$ and ending at each ancestor of node $u$; we use the notation $\mbox{sup}(p)$ to denote the set of super-paths of $p$. Observe that, the definition of $x$ implies that a super-path $q$ of $p$ that is longer than $p$ by $j$ has $x_q=\frac{1}{g-1}g^{1-j}x_p$ if $q$ is adjacent to the root and $x_q=g^{-j} x_p$ otherwise. Hence, assuming that node $u$ is at depth $\ell\geq 1$ from the root, we have that 
	
	\begin{eqnarray*}
		\sum_{p\in {\cal P}_e}{x_p}&=&\sum_{i=1}^g{\sum_{p\in \tilde{\cal P}_{e_i}}{\sum_{q\in \mbox{\tiny sup}(p)}{x_q}}}
		= \left(\sum_{j=1}^{\ell-1}{g^{-j}}+\frac{1}{g-1}g^{1-\ell}\right)\sum_{i=1}^g{\sum_{p\in \tilde{\cal P}_{e_i}}{x_p}}\\
		&=& \frac{1}{g-1}\left(\sum_{i=1}^g{\sum_{p\in {\cal P}_{e_i}}{x_p}}-\sum_{p\in {\cal P}_e}{x_p}\right), 
	\end{eqnarray*}
	which yields $\sum_{p\in {\cal P}_e}{x_p}=\frac{g}{g-2}$ as desired, since $\sum_{p\in {\cal P}_{e_i}}{x_p}=\frac{g}{g-2}$ by the induction hypothesis.
	
	It remains to show feasibility. Clearly, $\sum_{p\in {\cal P}_e}{x_p}=\frac{g}{g-2} \geq 1$ if $e$ is adjacent to the root. Otherwise, consider an edge $e$, its parent edge $f$, and their common endpoint $u$. Assuming that $u$ is at depth $\ell$ from the root (and using definitions and observations we used above), we have
	\begin{eqnarray*}
		\sum_{p\in {\cal P}_e\cap {\cal P}_f}{x_p} &=& \sum_{p\in \tilde{\cal P}_e}{\sum_{q\in \mbox{\tiny sup}(p)}{x_q}}
		= \left(\sum_{j=1}^{\ell-1}{g^{-j}}+\frac{1}{g-1}g^{1-\ell}\right)\sum_{p\in \tilde{\cal P}_e}{x_p}
		= \frac{1}{g-1}\sum_{p\in \tilde{\cal P}_e}{x_p},
	\end{eqnarray*}
	which, together with the fact that $\frac{g}{g-2}=\sum_{p\in{\cal P}_e}{x_p}=\sum_{p\in {\cal P}_e\cap {\cal P}_f}{x_p}+\sum_{p\in \tilde{\cal P}_e}{x_p}$ yields $\sum_{p\in {\cal P}_e\cap {\cal P}_f}{x_p}=\frac{1}{g-2}$ and $\sum_{p\in \tilde{\cal P}_e}{x_p}=\frac{g-1}{g-2}$ and, consequently, $\sum_{p\in \tilde{\cal P}_e}{x_p}-\sum_{p\in {\cal P}_e\cap {\cal P}_f}{x_p}=1$ as desired.  
\end{proof}

Finally, we state the lemma for augmentation factor $g=2$. 

\begin{lemma}\label{lem:g2}
	Let $T$ be an $N$-node $2$-tree. Then, there is a path covering of $T$ of cost at most $\log{N}$.
\end{lemma}

\begin{proof}
We will construct the path covering $x$ by visiting the edges of the tree in a bottom-up manner, i.e., first visiting edges that are incident to leaves and in such a way that an edge that is not adjacent to a leaf is visited only after its two child-edges have been visited. The assignment $x$ will be defined using a temporary assignment $y$. When visiting an edge $e$ that is adjacent to a leaf, we determine the temporary value $y_p=\log{N}$ associated with the path $p$ consisting of edge $e$ only. When visiting an edge $e$ that is not adjacent to a leaf, we set a temporary positive value $y_p$ for each path $p$ that begins with edge $e$ and we determine the final value $x_p$ for each path of $\tilde{\cal P}_{e_i}$ that begins with the child-edge $e_i$ (with $i\in\{1,2\}$) of $e$. In particular, let $p$ be a path of $\tilde{\cal P}_{e_i}$ and let $y_p$ be the temporary value assigned to it during our previous visit to edge $e_i$. After the phase associated with edge $e$, for the super-path $q$ of $p$ that begins with edge $e$, we set the temporary value 
$$y_q=\frac{\sum_{p'\in \tilde{\cal P}_{e_i}}{y_{p'}}-1}{2\sum_{p'\in \tilde{\cal P}_{e_i}}{y_{p'}}}y_p$$ and, determine the final value
$$x_p=\frac{\sum_{p'\in \tilde{\cal P}_{e_i}}{y_{p'}}+1}{2\sum_{p'\in \tilde{\cal P}_{e_i}}{y_{p'}}}y_p.$$
Let $f$ denote the parent of edge $e$ (if any). Observe that $x_p+y_q=y_p$ which means that the temporary value of a path $p$ in $\tilde{\cal P}_e$ is redistributed as final value of the path and temporary value of its super-path that begins with $f$. This argument can be repeated for all super-paths of $p$ and implies that the total temporary value of the paths in $\tilde{\cal P}_e$ is redistributed as total final value of path $p$ and the paths in $\mbox{sup}(p)$, i.e., $y_p=\sum_{q\in \{p\}\cup \mbox{\tiny sup}(p)}{x_p}$.

Also, observe that $x_p-y_q=\frac{y_p}{\sum_{p'\in \tilde{\cal P}_{e_i}}{y_{p'}}}$ after the phase associated with edge $e$ and, hence, $\sum_{p\in \tilde{\cal P}_{e_i}}{x_p} - \sum_{q\in {\cal P}_{e_i}\cap \tilde{\cal P}_e}{y_q}=1$. Since the temporary value $y_p$ on a path $p$ is redistributed as final value on $p$ and the paths of $\mbox{sup}(p)$, we have $\sum_{q\in {\cal P}_{e_i}\cap \tilde{\cal P}_e}{y_q}=\sum_{p\in {\cal P}_{e_i}\cap {\cal P}_e}{x_p}$ and the feasibility condition $\sum_{p\in \tilde{\cal P}_{e_i}}{x_p} - \sum_{p\in {\cal P}_{e_i}\cap {\cal P}_e}{x_p}=1$ on edge $e_i$ follows by the last two equalities.

We still have to prove the bound on the cost of $x$ as well as the feasibility condition for the edge adjacent to the root. In order to do so, we will prove that for every edge $e=(u,v)$ which defines a subtree with $N_e$ nodes (including both its endpoints), it holds that $\log{N}\geq \sum_{p\in {\cal P}_e}{x_p} \geq \log{N}-\log{N_e}+1$.

Since $N_e=2$ for every edge $e$ incident to a leaf, both inequalities hold (and are essentially the same equality) in this case. Now consider an edge $e$ that is not incident to a leaf and is such that $\log{N}\geq \sum_{p\in {\cal P}_{e_i}}{x_p} \geq \log{N}-\log{N_{e_i}}+1$ for each child-edge $e_i$ (with $i\in\{1,2\}$) of $e$. Using the feasibility condition for edges $e_1$ and $e_2$, we have
\begin{eqnarray}\nonumber
\sum_{p\in {\cal P}_e}{x_p} &=& \sum_{p\in {\cal P}_e\cap {\cal P}_{e_1}}{x_p} + \sum_{p\in {\cal P}_e\cap {\cal P}_{e_2}}{x_p}\\\nonumber
&=& \frac{1}{2}\sum_{p\in {\cal P}_{e_1}\cap {\cal P}_e}{x_p}+\frac{1}{2}\left(\sum_{p\in \tilde{\cal P}_{e_1}}{x_p}-1\right) 
+\frac{1}{2}\sum_{p\in {\cal P}_{e_2}\cap {\cal P}_e}{x_p}+\frac{1}{2}\left(\sum_{p\in \tilde{\cal P}_{e_2}}{x_p}-1\right)\\\label{eq:recursion}
&=& \frac{1}{2}\sum_{p\in {\cal P}_{e_1}}{x_p}+\frac{1}{2}\sum_{p\in {\cal P}_{e_2}}{x_p}-1.
\end{eqnarray}
Recall that $\sum_{p\in {\cal P}_e}{x_p}=\log{N}$ for every edge that is incident to a leaf. So, (\ref{eq:recursion}) implies that $\sum_{p\in {\cal P}_e}{x_p}\leq \log{N}$ for any edge $e$ as well, and the bound on the cost of $x$ follows. 

Using (\ref{eq:recursion}) and the assumption on the total final value of paths in ${\cal P}_{e_1}$ and ${\cal P}_{e_2}$, we get
\begin{eqnarray*}
	\sum_{p\in {\cal P}_e}{x_p} &\geq & \frac{1}{2}\left(\log{N}-\log{N_{e_1}}+1\right)
	+\frac{1}{2}\left(\log{N}-\log{N_{e_2}}+1\right)-1\\
	&=& \log{N} - \log{\sqrt{N_{e_1}\cdot N_{e_2}}}
	\geq  \log{N}-\log{\left(\frac{N_{e_1}+N_{e_2}}{2}\right)}\\
	&= & \log{N}-\log{N_e}+1.
\end{eqnarray*}
The second inequality follows by the relation of the geometric and arithmetic mean and the last equality is due to the fact that $N_e=N_{e_1}+N_{e_2}$. This completes the proof of the lemma.
\end{proof}

\subsection{The approximation ratio of Random Serial Dictatorship}

We have shown that the performance of SD significantly improves even with a small augmentation factor. A natural next question is to study its randomized counterpart, RSD. Could randomization help in achieving much better ratios? In the following, we prove an approximation guarantee for RSD, when there is no resource augmentation. 

\begin{theorem}\label{thm:RSD}
	The approximation ratio of RSD without resource augmentation is $ratio(RSD)\leq n$.
\end{theorem}

\begin{proof}
	Here again, we will use the alternative interpretation of the problem, where there are $n$ agents and $n$ items (capacity slots are interpreted as different items). We will prove the lemma by induction on $n$.
	
	When $n=1$, $RSD$ outputs the optimal solution. For the induction step, assume that for $n=k$, it holds that $ratio(RSD)\leq k$ and consider the case when $n=k+1$. Let $d_{i}=\min_{j\in [k+1]}d(A_i,F_j)$ and $t_i=\arg\min_{j\in [k+1]}d(A_i,F_j)$. In slight abuse of notation, let $SC_{RSD}(L,T)$ be the expected social cost of Random Priority on any instance with agents in $L$ and items in $T$. Similarly, let $SC_{OPT}(L,T)$ denote the social cost of the optimal assignment between agents in $L$ and items in $T$. Then, we have that
	
	\begin{equation*}\label{1}
	\begin{split}
	SC_{RSD}(N,M)
	& =    \frac{1}{k+1}\sum_{i=1}^{k+1}\left(d_i+SC_{RSD}(N\!-\!\{i\},M\!-\!\{t_i\})\right)                                   \\
	& \le  \frac{1}{k+1}\sum_{i=1}^{k+1}d_i+\frac{1}{k+1}\sum_{i=1}^{k+1}k\cdot SC_{OPT}(N\!-\!\{i\},M\!-\!\{t_i\})     \\
	& \le  \frac{1}{k+1}\sum_{i=1}^{k+1}d_i+\frac{1}{k+1}\sum_{i=1}^{k+1}k\cdot (SC_{OPT}(N,M)+d_i)     \\
	& \le  \frac{1}{k+1}\sum_{i=1}^{k+1}d_i+\frac{1}{k+1}\sum_{i=1}^{k+1}\left(k\cdot d_i+ k\cdot SC_{OPT}(N,M) \right)  \\
	& \le  (k+1)\cdot SC_{OPT}(N,M)
	\end{split}
	\end{equation*}
	where the first inequality follows from the induction hypothesis and the last inequality follows from the fact that $SC_{OPT}(N,M)\ge \sum_{i=1}^{k+1}d_i$. For the second inequality, observe first that if in the optimal assignment, agent $i$ is matched with $t_i$, then the inequality holds and we are done. Hence, assume without loss of generality that in the optimal assignment, agent $i$ is matched with some item $j$ and item $j^*$ is matched with some agent $i^*$. Then, if we remove agent $i$ and item $j^*$ from the optimal assignment on $N$ and $M$ and add the pair $i^*$ and $j$, we obtain an assignment on $N\!-\!\{i\},M\!-\!\{j^*\}$. Let $S$ be that assignment and let $SC_{S}(N\!-\!\{i\},M\!-\!\{j^*\})$ be its social cost. By the definition of $SC_{OPT}(N\!-\!\{i\},M\!-\!\{j^*\})$, we have
	
	\begin{equation*}
	\begin{split}
	SC_{OPT}(N\!-\!\{i\},M\!-\!\{j^*\})
	& \le SC_{S}(N\!-\!\{i\},M\!-\!\{j^*\})\\
	& \le SC_{OPT}(N,M)-d(A_{i^*},F_{j^*})-d(A_i,F_j)+d(A_{i^*},F_{j})\\
	& \le SC_{OPT}(N,M)+d(A_i,F_{j^*})
	\end{split}
	\end{equation*}
	where the last inequality follows from the triangle inequality. This completes the proof of the lemma.
	\end{proof}
	
	

\section{Lower bounds for Serial Dictatorship and Random Serial Dictatorship}\label{sec:lowerbounds}

In this section, we provide lower bounds on the approximation ratio with augmentation of the two mechanisms that we study. Interestingly, the constructed instances are all on a simple metric space, the \emph{real line metric}. For SD, the lower bounds that we prove 
show that our analysis in Section \ref{sec:analysis} is tight. 
For RSD and augmentation $g=1$, while the bound is not tight, it shows that even if there is a more involved analysis that potentially yields better upper bounds, it is not possible to obtain a much better approximation ratio and in particular, it is not possible to match the logarithmic approximation guarantee of SD with augmentation $g=2$. The lower bounds will be established by the following theorem.  

\begin{theorem}\label{thm:lb}
	The approximation ratio of Serial Dictatorship with augmentation factor $g$ in facility assignment instances with $n$ agents is \begin{enumerate}
		\item $ratio(SD)\geq 2^n-1$
		\item $ratio_2(SD)\geq\log{(n+1)}$ 
		\item $ratio_g(SD) \geq \frac{g}{g-2}-\delta$ for any $\delta>0$ when $g\geq 3$.
		\end{enumerate}
		 The approximation ratio of Random Serial Dictatorship is at least $ratio(RSD) \geq n^{0.26}$ (without resource augmentation).
\end{theorem}

All the statements in Theorem \ref{thm:lb} will follow by the same construction, with agents and facilities lying on the real line. Note that similar instances for proving the simplest cases of Theorem \ref{thm:lb} have appeared in the related literature in the past \cite{khuller1994line,meyerson2006randomized,kalyanasundaram2000online}; here we include those instances as part of a more general construction that allows us to obtain lower bounds for different augmentation factors as well as Random Serial Dictatorship.

 Let $k>0$ be a positive integer and $\epsilon>0$. There are $k+2$ points of interest that will host agents and facilities; these have the coordinates $-\epsilon$, $1$, $2$, ..., and $2^{k}$. For $i=0, 1, ..., k-1$, there are $\ell_i$ agents at level $i$ and are located at point $2^i$. For $i=0, 1, ..., k$, we use $n_i=\sum_{j=0}^i{\ell_i}$. Facilities are partitioned into $k+1$ levels; each level has a single facility. The facility of level $0$ has capacity $c_0$ and is located at point $-\epsilon$. For $i=1, 2, ..., k$, the facility of level $i$ is located at point $2^i$ and has capacity $c_i$. The different lower bounds will be obtained by setting the values of the quantities $\ell_i$ and $c_i$ appropriately. Note that the optimal cost is at most $\ell_0 (1+\epsilon)$ which is obtained by assigning the agents of level $i$ to the facility of level $i$ for $i=0, 1, ..., k$. Clearly, the optimal cost can become arbitrarily close to $\ell_0$ by selecting $\epsilon$ to be sufficiently small.

\paragraph{Proof of Statements (1), (2), and (3) in Theorem \ref{thm:lb}} We will set the parameters of the construction appropriately and will consider the execution of SD using any ordering of the agents that is {\em non-decreasing} in terms of level. Let $g\geq 1$ be the augmentation factor; the case $g=1$ will handle the no-augmentation case. We set $\ell_i=c_i=g^{k-i-1}$ for $i=0, 1, ..., k-1$ and $c_k=1$. Note that the $g^k-1$ agents of level $0$ that are considered first will be assigned to the facility of level $1$ which is their closest one; it is at distance $1$ from the agents of level $0$, clearly closer compared to facilities of higher levels but also closer compared to the facility of level $0$ which is at distance $1+\epsilon$ from the agents of level $0$. Note that the (augmented) capacity of the facility of level $1$ is exactly $g^{k-1}$ which means that the agents of level $0$ have occupied it in full. The agents of level $1$ appear next in the ordering and are assigned to facility of level $2$ (since it is the closest facility that has available space). Again, the agents of level $1$ occupy the facility in full. Continuing in this way, we have that the agents of level $i$ (located at point $2^i$) are assigned to the facility of level $i+1$ (at point $2^{i+1}$) for $i=0, 1, ..., k-1$. 

 \begin{figure}\label{fig:lb-instances}
 	\centering
 \begin{tikzpicture}[scale=0.3]
 
 %
 
 \draw (0,0) -- (35,0);
 
 
 \node[rectangle, inner sep=0pt, minimum size=0.5cm,draw=black, fill=lightgray] (F1) at (0,0) {$F_1$}; 
 \node[rectangle, inner sep=0pt, minimum size=0.5cm,draw=black, fill=lightgray] (F2) at (4,0) {$F_2$}; 
 \node[rectangle, inner sep=0pt, minimum size=0.5cm,draw=black, fill=lightgray] (F3) at (8,0) {$F_3$}; 
 \node[rectangle, inner sep=0pt, minimum size=0.5cm,draw=black, fill=lightgray] (F4) at (16,0) {$F_4$}; 
 \node[rectangle, inner sep=0pt, minimum size=0.5cm,draw=black, fill=lightgray] (F5) at (32,0) {$F_5$}; 
 
 \node[rectangle, inner sep=0pt, minimum size=0.3cm,draw=black] (A1) at (2,2) {}; 
 \node[rectangle, inner sep=0pt, minimum size=0.3cm,draw=black ] (A2) at (4,2) {}; 
 \node[rectangle, inner sep=0pt, minimum size=0.3cm,draw=black] (A3) at (8,2) {}; 
 \node[rectangle, inner sep=0pt, minimum size=0.3cm,draw=black] (A4) at (16,2) {}; 
 \node[rectangle, inner sep=0pt, minimum size=0.3cm,draw=black] (A5) at (32,2) {}; 
 
 \node[rectangle, inner sep=0pt, minimum size=0.3cm,draw=black] (A1) at (2,3) {}; 
 \node[rectangle, inner sep=0pt, minimum size=0.3cm,draw=black] (A2) at (4,3) {}; 
 \node[rectangle, inner sep=0pt, minimum size=0.3cm,draw=black] (A3) at (8,3) {}; 
 \node[rectangle, inner sep=0pt, minimum size=0.3cm,draw=black] (A4) at (16,3) {}; 
 \node[rectangle, inner sep=0pt, minimum size=0.3cm,draw=black] (A5) at (32,3) {}; 
 
 \node[rectangle, inner sep=0pt, minimum size=0.3cm,draw=black] (A1) at (2,4) {}; 
 \node[rectangle, inner sep=0pt, minimum size=0.3cm,draw=black] (A2) at (4,4) {}; 
 \node[rectangle, inner sep=0pt, minimum size=0.3cm,draw=black] (A3) at (8,4) {}; 
 \node[rectangle, inner sep=0pt, minimum size=0.3cm,draw=black] (A4) at (16,4) {}; 
 
 \node[rectangle, inner sep=0pt, minimum size=0.3cm,draw=black] (A1) at (2,5) {}; 
 \node[rectangle, inner sep=0pt, minimum size=0.3cm,draw=black] (A2) at (4,5) {}; 
 \node[rectangle, inner sep=0pt, minimum size=0.3cm,draw=black] (A3) at (8,5) {}; 
 
 \node[rectangle, inner sep=0pt, minimum size=0.3cm,draw=black] (A1) at (2,6) {}; 
 \node[rectangle, inner sep=0pt, minimum size=0.3cm,draw=black] (A2) at (4,6) {}; 
 
 \node[rectangle, inner sep=0pt, minimum size=0.3cm,draw=black] (A1) at (2,7) {}; 
 \node[rectangle, inner sep=0pt, minimum size=0.3cm,draw=black] (A2) at (4,7) {}; 
 
 \node[rectangle, inner sep=0pt, minimum size=0.3cm,draw=black] (A1) at (2,8) {}; 
 
 \node[rectangle, inner sep=0pt, minimum size=0.3cm,draw=black] (A1) at (2,9) {}; 
 
 
 \node[below] (F1) at (0,-2) {$-\epsilon$};
 \node[below] (F1) at (2,-2) {$1$};
 \node[below] (F2) at (4,-2) {$2$};
 \node[below] (F3) at (8,-2) {$4$};
 \node[below] (F4) at (16,-2) {$8$};
 \node[below] (F5) at (32,-2) {$16$};
 
 \node[below] (F1) at (0,-1) {$c_0$};
 \node[below] (F2) at (4,-1) {$c_1$};
 \node[below] (F3) at (8,-1) {$c_2$};
 \node[below] (F4) at (16,-1) {$c_3$};
 \node[below] (F5) at (32,-1) {$c_4$};
 
 \node[above] (F1) at (2,10) {$l_0$};
 \node[above] (F2) at (4,10) {$l_1$};
 \node[above] (F3) at (8,10) {$l_2$};
 \node[above] (F4) at (16,10) {$l_3$};
 \node[above] (F5) at (32,10) {$l_4$};

 \end{tikzpicture}
 	\caption{The lower bound construction of Theorem \ref{thm:lb} for $5$ facilities. The gray boxes correspond to facilities, the white boxes correspond to agents. For example, by setting $l_0=c_0=16$, $l_1=c_1=8$, $l_2=c_2=4$, $l_3=c_3=2$ and $l_4=c_4=1$ we obtain the instance for the lower bound when $g=2$.}
 \end{figure}
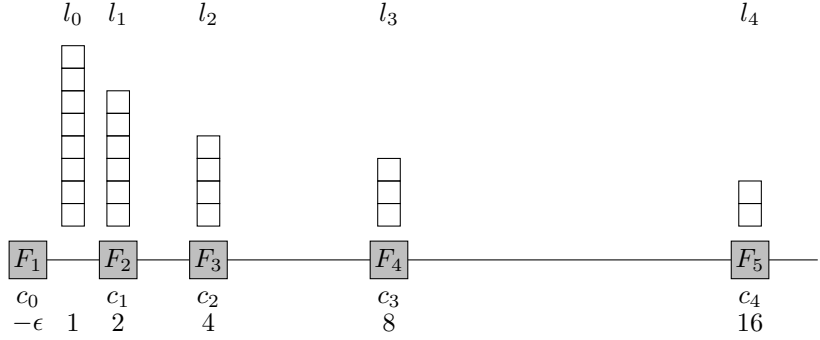

The social cost is then $\sum_{i=0}^{k-1}{g^{k-i-1} 2^i} = \ell_0 \sum_{i=0}^{k-1}{(2/g)^i}$ while the number of agents is $n=\sum_{i=0}^{k-1}{g^{k-i-1}}$. If $g=1$, we have $n=k$ agents and a social cost of $2^k-1=(2^n-1)\ell_0$. If $g=2$, we have $n=2^k-1$ and a social cost of $k \ell_0 = \ell_0 \log{(n+1)}$. Finally, for $g\geq 3$, we have a social cost of $g^{k-1}\frac{1-(2/g)^k}{1-2/g} \geq \ell_0 \left(\frac{g}{g-2}-\delta\right)$, where the inequality holds for every positive $\delta$ by selecting $k$ to be sufficiently large. This completes the proof of the first three statements.

For the lower bound of RSD, we set the parameters of the construction as follows: $\ell_0=c_0=1$ and $\ell_i=c_i=1+2\cdot 3^{i-1}$ for $i=1, ..., k-1$ and $c_k=1$. The proof is slightly more involved. We will need a definition and two technical lemmas. 

\begin{definition}
	An ordering has the ``chain of levels'' property if, for $i=1, 2, ..., k-1$, at least one agent of level $i$ appears after all agents of levels $0, 1, ..., i-1$.
\end{definition}

\begin{lemma}\label{lem:prob}
	The probability that a random ordering of the agents has a chain of levels is $\prod_{i=1}^{k-1}{\left(1-\frac{n_{i-1}}{n_i}\right)}$.
\end{lemma}

\begin{proof}
	We can view the generation of a uniformly random ordering of all agents as a process that proceeds level by level. At level $0$, the process simply computes a uniformly random ordering of the agents of level $0$. At level $i>0$, it computes a uniformly random ordering of the agents in levels $0, 1, ..., i$ as follows. It uses the random ordering of the agents in levels $0, 1, ..., i-1$, computes a uniformly random ordering of the agents of level $i$ and picks one among the possible merges of the two orderings uniformly at random. 
	
	Now, in each step $i>0$, the number of possible merges of the two orderings is equal to ${n_i\choose n_{i-1}}$ while the number of merged orderings in which the last agent belongs to level $i$ (as the ``chain of levels'' property requires) is ${n_i-1\choose n_{i-1}}$. Since the random events at the different steps are independent, we obtain that the probability that the resulting ordering will have a chain of levels is $\prod_{i=1}^{k-1}{{n_i-1\choose n_{i-1}}/{n_i\choose n_{i-1}}} = \prod_{i=1}^{k-1}{\left(1-\frac{n_{i-1}}{n_i}\right)}$.
\end{proof}

\begin{lemma}\label{lem:bad}
	Consider the application of SD on the above instance using an ordering of the agents that has a chain of levels. Then, for $i=0, 1, ..., k-1$, at least one agent of level $i$ is assigned to the facility of level $i+1$.
\end{lemma}

\begin{proof}
We first claim that an agent of level $i$ cannot be assigned to a facility of a lower level than $i$. Assume otherwise and consider the first agent $a$ in the ordering which belongs to level $i$ and is assigned to the facility of level $i'<i$. This means that all facilities in levels $i'+1, ..., i$ are full by agents that appear before agent $a$ in the ordering. Note that these facilities cannot contain any agent from levels higher than $i$ (since agent $a$ is the first one that is assigned to a facility of a lower level) or any agent in level $i'$ or lower (since the facility of level $i'$ which is closer to them has free space). Since agent $a$ belongs to level $i$ as well, we obtain that the total capacity of the facilities in levels $i'+1, ..., i$ is strictly smaller than the total number of agents in these levels; this contradicts the definition of the instance.

Now, we will prove the lemma by considering the last agent from each level in an ordering with a chain of levels. When the agent of level $0$ is considered, the chain of levels property guarantees that some of the agents of level $1$ has not appeared yet. Furthermore, the fact that no agent is ever assigned to a facility of lower level implies that the facility of level $1$ (which is closer to the agent compared to the facility of level $0$) has free space and the agent of level $0$ will be assigned to it. Now, consider the last agent of level $1$. When it appears, the facility of level $1$ is full; it contains the agent of level $0$ and the agents of level $1$ before the last one. Again, the chain of levels property guarantees that that some of the agents of level $2$ has not appeared yet. Together with the fact that no agent is ever assigned to a facility of lower level, this leads again to the conclusion that the facility of level $2$ (which is again closer to the agent compared to the facility of level $0$ which is still empty) has free space and the agent of level $0$ will be assigned to it. Continuing this reasoning completes the proof of the lemma.
\end{proof}

We now complete the proof as follows. Observe that the parameters are such that $n_0=1$, $n_i=1+\sum_{j=1}^i{3^{j-1}}=3^i$ for $i=1, ..., k-1$, and the number of agents is $n=n_{k-1}=3^{k-1}$. By Lemma \ref{lem:bad}, we have that if the random ordering used by RSD happens to have a chain of levels, then some agent of level $i$ will be assigned to the facility of level $i+1$, for $i=0, 1, ..., k-1$. The social cost in this case is $2^k-1\geq 2^{k-1}$. By Lemma \ref{lem:prob}, the probability that a random ordering has a chain of level is $(2/3)^{k-1}$. Hence, the expected social cost of RSD is at least $(4/3)^{k-1}=n^{\log_3{(4/3)}} \approx \ell_0 n^{0.26}$ and the bound for RSD follows.

Again, we have the following corollary.

\begin{corollary}
	For the online transportation problem with adversarial arrivals, the double-competitive ratio of Greedy is at least $\log (n+1)$ and the $g$-competitive ratio of Greedy is at least $g/(g-2)-\delta$ for any $\delta >0$, for $g \geq 3$. For the online transportation problem with uniform random arrivals, the competitive ratio of Greedy is at least $n^{0.29}$.
\end{corollary}	

\section{Discussion}
We proposed the employment of a resource augmentation framework for algorithmic mechanism design, where a mechanism, severely limited by the need for truthfulness (or any other desired property for that matter) is given some additional allocative power before being compared to the omnipotent optimal mechanism, which operates under no restrictions. We applied this framework to a natural problem that we call \emph{the facility assignment problem} and proved that a very-well known mechanism, Serial Dictatorship, while often being critized for its bad worst-case guarantees, is actually quite efficient if we allow even limited augmentation. 


The resouce augmentation framework is applicable to other problems in algorithmic mechanism design as well. For example, for the problem of one-sided matchings under general cardinal utilities \cite{HZ:79}, given the strong negative inapproximability bounds on the social welfare given in \cite{FFZ:14}, one could suspect that a resource augmentation approach, where we are allowed to make limited copies of the items, could provide much better approximation guarantees. Similarly, resouce augmentation could find applications in the wide literature on \emph{auctions} where the performance of a truthful mechanism with respect to some objective (e.g. revenue) would be measured against the optimal mechanism which runs on a smaller set of items. 

In fact, the framework can be applied to broader settings where the loss in performance is due to restrictions other than truthfulness, such as fairness \cite{caragiannis2012efficiency,aumann2010efficiency}, stability \cite{anshelevich2009anarchy,emek2015price} or ordinality \cite{FFZ:14,aziz2015egalitarianism}; all the problems in those papers can be studied through the resource augmentation lens. It is not hard to imagine that similar well-known notions like the \emph{Price of Fairness} \cite{caragiannis2012efficiency}, could be redefined to incorporate the possibility of resource augmentation.


Regarding the facility assignment problem, there are some interesting open questions to be answered. We took a positive first step in the study of Random Serial Dictatorship, proving approximation ratio bounds when there is no augmentation. It seems natural to explore what happens when $g\geq2$ for RSD; there does not seem to be a clear way to adapt the lower bound for $g=1$ to work for the case when $g=2$ and on the other hand, modifying the inductive argument of Theorem \ref{thm:RSD} to work for $g=2$ is not straightforward either. It seems like an interesting technical question to obtain (tight) bounds for RSD and for different augmentation factors. It would also be meaningful to consider augmentation factors smaller than $1$; note that a similar construction to the one in our main lower bound can be used to show that additive factors can not achieve significantly improved approximations. 

Finally, it makes sense to consider other truthful mechanisms, beyond the greedy ones, for the facility assignment problem. One could imagine that mechanisms that somehow take into account the numerical values of the distances between facilities and preferred locations, rather than simply the ordering of facilities induced by those distances could potentially outperform the greedy mechanisms above. However, such mechanisms that are also truthful are hard to construct and more importantly, since we are dealing with costs, they have to be \emph{unanimous}, i.e., if there is an assignment of zero social cost, it has to be outputted with probability 1. This immediately precludes using a straightforward adaptation of the mechanisms proposed in \cite{FT:10} (see also the full version of \cite{FFZ:14} for some examples).

Attempting to explore the limitations of truthful mechanisms, one could try to identify some of their structural characteristics, but full characterizations do not exist even for assignments under general preferences.\footnote{The closest thing that we have is the characterization by \cite{SVE:99} for a large class of truthful deterministic mechanisms and the one by \cite{mennle2014axiomatic} for randomized, ordinal truthful mechanisms, but since our setting is more restricted, it is not clear that those characterizations extend either.} Alternatively, we can try to prove lower bounds using only truthfulness as a property. The limitations of all truthful mechanisms for $2$ facilities and arbitrary capacities are settled in Section \ref{sec:twofacilities} of the Appendix, following this approach, where we prove that $SD$ is optimal among all truthful mechanisms for the problem, even randomized ones.
\clearpage
\bibliographystyle{plain}
\bibliography{resource}

\begin{thebibliography}{10}

\bibitem{abdulkadirouglu1998random}
A.~Abdulkadiro{\u{g}}lu and T.~S{\"o}nmez.
\newblock {Random serial dictatorship and the core from random endowments in
  house allocation problems}.
\newblock {\em Econometrica}, pages 689--701, 1998.

\bibitem{ACV15}
F.~Abed, I.~Caragiannis, and A.~A. Voudouris.
\newblock Near-optimal asymmetric binary matrix partitions.
\newblock In {\em Proceedings of the 40th International Symposium on
  Mathematical Foundations of Computer Science ({MFCS})}, pages 1--13, 2015.

\bibitem{anshelevich2010matching}
E.~Anshelevich and S.~Das.
\newblock {Matching, cardinal utility, and social welfare}.
\newblock {\em ACM SIGECom Exchanges}, 9(1):4, 2010.

\bibitem{anshelevich2009anarchy}
E.~Anshelevich, S.~Das, and Y.~Naamad.
\newblock Anarchy, stability, and utopia: Creating better matchings.
\newblock {\em Autonomous Agents and Multi-Agent Systems}, 26(1):120--140,
  2013.

\bibitem{anshelevich2015blind}
Elliot Anshelevich and Shreyas Sekar.
\newblock Blind, greedy, and random: Algorithms for matching and clustering
  using only ordinal information.
\newblock In {\em Proceedings of the 30th AAAI Conference on Artificial
  Intelligence (AAAI)}, pages 390--396, 2016.

\bibitem{antoniadis2014n}
Antonios Antoniadis, Neal Barcelo, Michael Nugent, Kirk Pruhs, and Michele
  Scquizzato.
\newblock A o (n)-competitive deterministic algorithm for online matching on a
  line.
\newblock In {\em International Workshop on Approximation and Online
  Algorithms}, pages 11--22. Springer, 2014.

\bibitem{ACK09}
S.~Athanassopoulos, I.~Caragiannis, and C.~Kaklamanis.
\newblock Analysis of approximation algorithms for \emph{k}-set cover using
  factor-revealing linear programs.
\newblock {\em Theory of Computing Systems}, 45(3):555--576, 2009.

\bibitem{aumann2010efficiency}
Y.~Aumann and Y.~Dombb.
\newblock The efficiency of fair division with connected pieces.
\newblock {\em ACM Transactions on Economics and Computation}, 3(4):art.~23,
  2015.

\bibitem{aziz2015egalitarianism}
H.~Aziz, J.~Chen, A.~Filos-Ratsikas, S.~Mackenzie, and N.~Mattei.
\newblock Egalitarianism of random assignment mechanisms.
\newblock In {\em Proceedings of the 10th International Conference on
  Autonomous Agents and Multiagent Systems (AAMAS)}, 2016.

\bibitem{BM:01}
A.~Bogomolnaia and H.~Moulin.
\newblock {A new solution to the random assignment problem}.
\newblock {\em Journal of Economic Theory}, 100:295--328, 2001.

\bibitem{C09}
I.~Caragiannis.
\newblock Wavelength management in {WDM} rings to maximize the number of
  connections.
\newblock {\em SIAM Journal on Discrete Mathematics}, 23(2):959--978, 2009.

\bibitem{caragiannis2012efficiency}
I.~Caragiannis, C.~Kaklamanis, P.~Kanellopoulos, and M.~Kyropoulou.
\newblock The efficiency of fair division.
\newblock {\em Theory of Computing Systems}, 50(4):589--610, 2012.

\bibitem{CS14}
D.~Chakrabarty and C.~Swamy.
\newblock {Welfare maximization and truthfulness in mechanism design with
  ordinal preferences}.
\newblock In {\em Proceedings of the 5th Conference on Innovations in
  Theoretical Computer Science (ITCS)}, pages 105--120, 2014.

\bibitem{chung2008online}
C.~Chung, K.~Pruhs, and P.~Uthaisombut.
\newblock The online transportation problem: On the exponential boost of one
  extra server.
\newblock In {\em Proceedings of the 8th Latin American Symposium on
  Theoretical Informatics (LATIN)}, pages 228--239. 2008.

\bibitem{emek2015price}
Y.~Emek, T.~Langner, and R.~Wattenhofer.
\newblock The price of matching with metric preferences.
\newblock In {\em Proceedings of the 23rd Annual European Symposium on
  Algorithms (ESA)}, pages 459--470. 2015.

\bibitem{FJ15}
U.~Feige and S.~Jozeph.
\newblock Oblivious algorithms for the maximum directed cut problem.
\newblock {\em Algorithmica}, 71(2):409--428, 2015.

\bibitem{FT:10}
U.~Feige and M.~Tennenholtz.
\newblock Responsive lotteries.
\newblock In {\em Proceedings of the 3rd International Symposium on Algorithmic
  Game Theory (SAGT)}, pages 150--161. 2010.

\bibitem{FFZ:14}
A.~Filos-Ratsikas, S.~K.~S. Frederiksen, and J.~Zhang.
\newblock Social welfare in one-sided matchings: Random priority and beyond.
\newblock In {\em Proceedings of the 7th International Symposium on Algorithmic
  Game Theory (SAGT)}, pages 1--12, 2014.

\bibitem{FM:13}
A.~Filos-Ratsikas and P.~B. Miltersen.
\newblock Truthful approximations to range voting.
\newblock In {\em Proceedings of the 10th International Conference on Web and
  Internet Economics (WINE)}, pages 175--188, 2014.

\bibitem{GC:10}
M.~Guo and V.~Conitzer.
\newblock {Strategy-proof allocation of multiple items between two agents
  without payments or priors}.
\newblock In {\em Proceedings of the 9th International Conference on Autonomous
  Agents and Multiagent Systems (AAMAS)}, pages 881--888, 2010.

\bibitem{HZ:79}
A.~Hylland and R.~Zeckhauser.
\newblock {The efficient allocation of individuals to positions}.
\newblock {\em The Journal of Political Economy}, 87(2):293--314, 1979.

\bibitem{JMM+03}
K.~Jain, M.~Mahdian, E.~Markakis, A.~Saberi, and V.~V. Vazirani.
\newblock Greedy facility location algorithms analyzed using dual fitting with
  factor-revealing {LP}.
\newblock {\em Journal of the ACM}, 50(6):795--824, 2003.

\bibitem{kalyanasundaram1993online}
B.~Kalyanasundaram and K.~Pruhs.
\newblock Online weighted matching.
\newblock {\em Journal of Algorithms}, 14(3):478--488, 1993.

\bibitem{kalyanasundaram2000online}
B.~Kalyanasundaram and K.~Pruhs.
\newblock The online transportation problem.
\newblock {\em SIAM Journal on Discrete Mathematics}, 13(3):370--383, 2000.

\bibitem{kalyanasundaram2000speed}
B.~Kalyanasundaram and K.~Pruhs.
\newblock Speed is as powerful as clairvoyance.
\newblock {\em Journal of the ACM}, 47(4):617--643, 2000.

\bibitem{khuller1994line}
S.~Khuller, S.~G. Mitchell, and V.~V. Vazirani.
\newblock On-line algorithms for weighted bipartite matching and stable
  marriages.
\newblock {\em Theoretical Computer Science}, 127(2):255--267, 1994.

\bibitem{koutsoupias1999weak}
E.~Koutsoupias.
\newblock Weak adversaries for the k-server problem.
\newblock In {\em Proceedings of the 40th Annual Symposium on Foundations of
  Computer Science (FOCS)}, pages 444--449, 1999.

\bibitem{koutsoupias2003online}
Elias Koutsoupias and Akash Nanavati.
\newblock The online matching problem on a line.
\newblock In {\em International Workshop on Approximation and Online
  Algorithms}, pages 179--191. Springer, 2003.

\bibitem{krysta2014size}
P.~Krysta, D.~Manlove, B.~Rastegari, and J.~Zhang.
\newblock {Size versus truthfulness in the House Allocation problem}.
\newblock In {\em Proceedings of the 15th ACM Conference on Economics and
  Computation (EC)}, pages 453--470, 2014.

\bibitem{MY11}
M.~Mahdian and Q.~Yan.
\newblock Online bipartite matching with random arrivals: an approach based on
  strongly factor-revealing lps.
\newblock In {\em Proceedings of the 43rd {ACM} Symposium on Theory of
  Computing ({STOC})}, pages 597--606, 2011.

\bibitem{mennle2014axiomatic}
T.~Mennle and S.~Seuken.
\newblock {An axiomatic approach to characterizing and relaxing
  strategyproofness of one-sided matching mechanisms}.
\newblock In {\em Proceedings of the 15th ACM Conference on Economics and
  Computation (EC)}, pages 37--38, 2014.

\bibitem{meyerson2006randomized}
A.~Meyerson, A.~Nanavati, and L.~Poplawski.
\newblock Randomized online algorithms for minimum metric bipartite matching.
\newblock In {\em Proceedings of the 17th Annual ACM-SIAM Symposium on Discrete
  Algorithms (SODA)}, pages 954--959, 2006.

\bibitem{PT:09}
A.~D. Procaccia and M.~Tennenholtz.
\newblock {Approximate mechanism design without money}.
\newblock {\em ACM Transactions on Economics and Computation}, 1(4):art.~18,
  2013.

\bibitem{sleator1985amortized}
D.~D. Sleator and R.~E. Tarjan.
\newblock Amortized efficiency of list update and paging rules.
\newblock {\em Communications of the ACM}, 28(2):202--208, 1985.

\bibitem{SVE:99}
L.-G. Svensson.
\newblock {Strategy-proof allocation of indivisble goods}.
\newblock {\em Social Choice and Welfare}, 16(4):557--567, 1999.

\bibitem{young1994thek}
N.~Young.
\newblock The k-server dual and loose competitiveness for paging.
\newblock {\em Algorithmica}, 11(6):525--541, 1994.

\end{thebibliography}

\clearpage
\appendix
\section*{Appendix}

\section{The online transportation problem} \label{sec:onlinetransportation}


As we mentioned earlier, there is a connection between the facility assignment problem and the \emph{online transportation problem} \cite{kalyanasundaram2000online} (also known as the minimum online metric bipartite matching or simply online metric matching \cite{meyerson2006randomized,koutsoupias2003online,khuller1994line}). 
In the \emph{online transportation problem}, there is a set of points $F$ on a metric space and a set of points $A$ that arrive in an online fashion. At each time that a point in $A$ arrives, it has to be matched to a point in $F$. The performance of an online algorithm is measured by its \emph{competitive ratio}, i.e., the worst-case ratio over all inputs of the social cost of the algorithm over the social cost of the optimal matching, that knows the exact sequence of arriving points in advance. Our setting can be interpreted as a similar metric matching problem, by ``splitting'' facilities with capacity $c_i>1$ to facilities of unit capacity that coincide on the metric space and by interpreting facilities as single, indivisible objects. Given this interpretation, SD and RSD can be thought of as greedy algorithms for the problem above. In particular, SD corresponds to the greedy algorithm in the setting with adversarial arrivals and RSD corresponds to the greedy algorithm when points in $A$ arrive uniformly at random. 

For the online transportation problem, it was known since the early 90s that without augmentation, Greedy achieves a competitive ratio of $2^{n}-1$ \cite{kalyanasundaram1993online}. Later on, Kalyanasundaram and Pruhs \cite{kalyanasundaram2000online} proved that when the online algorithm operates on doubled capacities, Greedy is $\Theta(\log n)$-competitive; given the discussion above, this implies a $\Theta(\log n)$-approximation bound for SD with $g=2$ in our setting. Note however that unlike the result in \cite{kalyanasundaram2000online}, our analysis is \emph{exact}, i.e., our $\log (n+1)$ bound involves no asymptotics.
Furthermore, we extend the result by proving exact bounds for any augmentation factor $g\geq 3$; the bounds are all small constants and in fact the ratio goes to $1$ as the augmentation factor grows large. 
These results naturally extend to the online transportation problem and confirm a conjecture by Chung et al. \cite{chung2008online}, namely that a constant competitive ratio can be achieved with augmentation factor $3$. 
Our results for RSD also imply upper and lower bounds for the performance of Greedy in the online transportation problem with uniform random arrivals. Specifically, Theorem \ref{thm:SD} and Theorem \ref{thm:RSD} give rise to the following corollary. We state the corollary using the terminology of the online problem \cite{kalyanasundaram2000online} for consistency.

\begin{corollary}
	The double-competitive ratio of \emph{Greedy} for the online transportation problem is at most $\log (n+1)$. The $g$-competitive ratio of \emph{Greedy} is at most $g/(g-2)$. The competitive ratio of \emph{Greedy} for the online transportation problem with uniform random arrivals is at most $n$. 
\end{corollary}	
Compared to the related result in \cite{kalyanasundaram2000online}, we remark that our analysis is substantially different due to the use of linear programming. This technique for the analysis of purely combinatorial algorithms has found applications in many different contexts such as facility location \cite{JMM+03}, set cover \cite{ACK09}, online matching \cite{MY11}, maximum directed cut \cite{FJ15}, wavelength routing \cite{C09}, and revenue optimization \cite{ACV15}. Like in our case, these techniques usually lead to tight analysis. Also note that while the connection between SD and RSD and the greedy algorithm for the online transportation problem is straightforward, the two problems are fundamentally different and hence non-greedy online competitive algorithms do not imply any bounds for our setting and non-serial truthful mechanisms do not imply any bounds for the online setting. 

\section{Facility assignment for two facilities} \label{sec:twofacilities}

In this section, we settle the question of truthful mechanisms for two facilities and arbitrary capacities, when there is no resource augmentation. Recall that for $m=2$, the approximation ratio of SD is $3$, and hence it suffices to prove that no truthful-in-expectation mechanism can achieve a better ratio. 

\begin{theorem}\label{thm:proof-of-concept}
	Let $M$ be any truthful-in-expectation mechanism and let $m=2$. Then, $ratio(M) \geq 3$.
\end{theorem}

\begin{proof}
	First, we claim that if there is a truthful-in-expectation mechanism $M$ such that $ratio_g(M)=\alpha$, then, there exists a truthful-in-expectation, anonymous mechanism $M'$ such that $ratio_g(M')\leq \alpha$ for any augmentation factor $g$. To see this, let Let $M'$ be the mechanism that given any instance $I_g$ applies a uniformly random permutation to the set of indices of the agents and then applies $M$ on $I_g$. The mechanism is clearly anonymous. Furthermore, since $I_g$ is a valid input to $M$, the approximation ratio of $M'$ can not be worse than that of $M$, since the approximation ratio is calculated over all possible input instances. For the same reason, if $M$ is truthful and since the permutation is independent of the reports, $M'$ is truthful-in-expectation. Similar arguments have been used before to prove similar statements in other contexts \cite{FFZ:14,FM:13,GC:10}.
	
	Now let $F_1$ and $F_2$ be the positions of the two facilities and let $|F_1-F_2| =2+\epsilon$. Let $c_1=n-1$ and $c_2=1$ and let $I=(A,F)$ be an instance such that $A_1=\ldots=A_{n-1}=F_1$ and $A_n=F_1+1$. By the discussion above we can without loss of generality assume that $M$ is anonymous. Let $p_n(A)$ be the probability that agent $n$ is assigned to facility $2$ in $A$. Next consider the profile $A'=(A_1,\ldots,A_{n-1},A'_n)$ with $A'_n = F_1$. By anonymity, the probability that agent $n$ is assigned to facility $2$ is $1/n$. Since agent $n$ prefers facility $1$ to facility $2$, truthfulness implies that on instance $A$, it must hold that $p_n(A) \leq 1/n$, otherwise agent $n$ would have an incentive to mireport $A'_n$ instead of $A_n$. This implies that on instance $A$, the expected welfare of $M$ is at most $\frac{1}{n}(1+\epsilon)+\frac{n-1}{n} (3+\epsilon)$ while the optimal social welfare is $1+\epsilon$. As $\epsilon \rightarrow 0$ and $n \rightarrow \infty$, the ratio goes to $3$.  
\end{proof}

\end{document}